\documentclass[11pt]{article}
\usepackage{times}
\usepackage{xspace,xcolor}
\usepackage{bbm}
\usepackage[shortlabels,inline]{enumitem}
\usepackage{amssymb,amsmath,amsthm,amsfonts,amstext}
\usepackage[margin=1cm]{caption}
\usepackage[shortlabels,inline]{enumitem}
\usepackage{hyperref}
\hypersetup{colorlinks = true,linkcolor = blue, anchorcolor = blue, citecolor = blue, filecolor = red,urlcolor = black}

\newtheorem{theorem}{Theorem}[section]
\newtheorem{lemma}[theorem]{Lemma}

{\theoremstyle{definition} 
\newtheorem{definition}[theorem]{Definition}
{\theoremstyle{remark} \newtheorem{remark}{Remark}}

\newtheoremstyle{empty}
        {\topsep}{\topsep}              %%% space between body and thm
        {\itshape}                      %%% Thm body font
        {}                              %%% Indent amount (empty = no indent)
        {\bfseries}                     %%% Thm head font
        {.}                             %%% Punctuation after thm head
        { }                             %%% Space after thm head
        {\thmnote{\bfseries #3}}       %%% Thm head spec
{\theoremstyle{empty} }

\def\blksquare{\rule{2mm}{2mm}}
\def\qedsymbol{\blksquare}

\newenvironment{proofof}[1]{\begin{proof}[Proof of #1]}{\end{proof}}

\newcommand{\bg}[1]{\medskip\noindent{\it #1}}

\makeatletter
\newcommand{\leqnomode}{\tagsleft@true\let\veqno\@@leqno}
\newcommand{\reqnomode}{\tagsleft@false\let\veqno\@@eqno}
\makeatother

\newcommand{\poly}{\operatorname{poly}}

\DeclareMathOperator*{\Exp}{E}
\newcommand{\prob}[2][{}]{\ensuremath{{\textstyle{\boldsymbol \Pr}_{#1}}[#2]}} %% fixed size probability
\newcommand{\probb}[2][{}]{\ensuremath{{\textstyle{\boldsymbol \Pr}_{#1}}\bigl[#2\bigr]}} %% adjustable probability
\newcommand{\E}[2][{}]{\ensuremath{{\textstyle{\boldsymbol \Exp}_{#1}}[#2]}}
\newcommand{\R}{\ensuremath{\mathbb R}}
\newcommand{\Rp}{\ensuremath{\mathbb R_{\geq 0}}}

\newcommand{\ceil}[1]{\ensuremath{\lceil#1\rceil}}

\newcommand{\floor}[1]{\ensuremath{\lfloor#1\rfloor}}

\newcommand{\pr}[1]{\ensuremath{\prob{#1}}}
\newcommand{\scal}{\theta}
\newcommand{\da}{\downarrow}

\newcommand{\mfont}{\mathsf}
\newcommand{\mc}{\mathcal}

\newcommand{\A}{\ensuremath{\mc{A}}}
\newcommand{\al}{\ensuremath{\alpha}}
\newcommand{\kp}{\ensuremath{\kappa}}
\newcommand{\ld}{\ensuremath{\lambda}}
\newcommand{\gm}{\ensuremath{\gamma}}
\newcommand{\eps}{\ensuremath{\epsilon}}
\newcommand{\sg}{{\ensuremath{\sigma}}}

\newcommand{\lvec}{\ensuremath{\overrightarrow{\load}}}
\newcommand{\lvecsg}[1][\sg]{\lvec^{{#1}}}

\newcommand{\scrv}{\ensuremath{Z}}
\newcommand{\vecrv}{\ensuremath{W}}
\newcommand{\sumrv}{\ensuremath{S}}

\newcommand{\boldone}{\ensuremath{\mathbbm{1}}}
\newcommand{\msnorm}{\ensuremath{g}}
\newcommand{\dl}{\ensuremath{\delta}}

\newcommand{\load}{\ensuremath{\mfont{load}}}
\newcommand{\OPT}{\ensuremath{\mfont{OPT}}}
\newcommand{\hi}{\ensuremath{\mfont{hi}}}
\newcommand{\low}{\ensuremath{\mfont{low}}}
\newcommand{\LB}{\ensuremath{\mathit{lb}}}

\newcommand{\poislb}{\ensuremath{\mfont{PoisNormLB}}\xspace}
\newcommand{\stochlb}{\ensuremath{\mfont{StochNormLB}}\xspace}
\newcommand{\stochflb}[1][$f$]{\ensuremath{\mfont{Stoch}}\text{-}{#1}\text{-}\ensuremath{\mfont{LB}}\xspace}

\newcommand{\POS}{\mfont{POS}}

\newcommand{\erv}[2]{\ensuremath{#1^{\geq #2}}} % exceptional random variable
\newcommand{\trv}[2]{\ensuremath{#1^{< #2}}} % truncated random variable
\newcommand{\eff}[2][\lambda]{\beta_{#1}(#2)} % fixed effective size
\newcommand{\topl}[1][\ell]{\ensuremath{\mfont{Top}_{#1}}}
\newcommand{\Topl}[2][\ell]{\ensuremath{\topl[{#1}](#2)}}

\title{A Simple Approximation Algorithm for Vector Scheduling \\ and Applications to Stochastic Min-Norm Load Balancing\thanks{An extended abstract is to appear in the Proceedings of the 5th SOSA, 2022.}} 
\author{
    Sharat Ibrahimpur\thanks{{\tt \{sharat.ibrahimpur,cswamy\}@uwaterloo.ca}.
    Department of Combinatorics and Optimization, University of Waterloo, Waterloo, ON, Canada N2L 3G1. 
    Supported in part by NSERC grant 327620-09 and an NSERC Discovery Accelerator
    Supplement Award.}
\and 
\addtocounter{footnote}{-1} 
    Chaitanya Swamy\footnotemark
}
\date{}

\begin{document}

\maketitle	
	
\begin{abstract}

We consider the \emph{Vector Scheduling} problem on identical machines: we have $m$ machines, and a set $J$ of $n$ jobs, where each job $j$ has a processing-time \emph{vector} $p_j \in \Rp^d$.
The goal is to find an assignment $\sg : J \to [m]$ of jobs to machines so as to minimize the \emph{makespan} $\max_{i \in [m]} \max_{r \in [d]} \bigl( \sum_{j : \sg(j) = i} p_{j,r} \bigr)$. 
A natural lower bound on the optimal makespan is $\LB := \max \{\max_{j \in J,r \in [d]} p_{j,r}, \max_{r \in [d]} (\sum_{j \in J} p_{j,r} / m) \}$.
Our main result is a \emph{very simple} $O(\log d)$-approximation algorithm for vector scheduling with respect to the lower bound $\LB$: we devise an algorithm that returns an assignment whose makespan is at most $O(\log d) \cdot \LB$.

\smallskip

As an application, we show that the above guarantee leads to an $O(\log \log m)$-approximation for \emph{stochastic minimum-norm load balancing} ($\stochlb$).
In {\stochlb}, we have $m$ identical machines, a set $J$ of $n$ independent stochastic jobs whose processing times are nonnegative random variables, and a monotone, symmetric norm $f : \R^m \to \Rp$.
The goal is to find an assignment $\sg : J \to [m]$ that minimizes the \emph{expected $f$-norm} of the machine-load vector $\lvecsg$, where the $i^\mathrm{th}$ coordinate of $\lvecsg$ is the (random) total processing time assigned to machine $i$.
Our $O(\log \log m)$-approximation guarantee is in fact much stronger: we obtain an assignment that is {\em simultaneously} an $O(\log \log m)$-approximation for \stochlb with {\em all} monotone, symmetric norms. 
Next, this approximation factor significantly improves upon the $O(\log m/\log \log m)$-approximation in (Ibrahimpur and Swamy, FOCS 2020) for \stochlb, and is a consequence of a more-general \emph{black-box} reduction that we present, showing that a $\gm(d)$-approximation for $d$-dimensional vector scheduling \emph{with respect to the lower bound $\LB$} yields a simultaneous $\gm(\log m)$-approximation for \stochlb with all monotone, symmetric norms.
We emphasize that it is crucial for this reduction that the approximation guarantee for vector scheduling is with respect to the lower bound $\LB$.  

\end{abstract}

\section{Introduction} \label{sec:intro}

We consider the well-studied \emph{Vector Scheduling} problem on identical machines, a natural generalization of the (scalar) makespan minimization problem to a setting with $d \geq 1$ dimensions.
We have a set $J$ of $n$ jobs that are to be processed on exactly one of $m$ identical machines.
We use $[N]$ to denote $\{1,\dots,N\}$. 
Each job $j \in J$ has a processing-time (or size) \emph{vector} $p_j = (p_{j,1},\dots,p_{j,d}) \in \Rp^d$, where $p_{j,r}$ denotes the size of job $j$ in dimension $r \in [d]$; for example, dimensions could correspond to processor cycles, storage space, or network bandwidth needed to complete the job.
An assignment $\sg : J \to [m]$ of jobs to machines induces an $(m \times d)$-dimensional load vector $\lvecsg$ (one entry per machine-dimension pair): for machine $i \in [m]$ and dimension $r \in [d]$, the load in the $r^\mathrm{th}$ dimension of machine $i$ is $\lvec_{i,r} := \sum_{j : \sg(j) = i} p_{j,r}$.
The \emph{makespan objective} of an assignment $\sg$ is defined as $\max_{i \in [m], r \in [d]} \lvec_{i,r}$ i.e., the maximum load across all machines and all dimensions.
The goal in vector scheduling is to find an assignment $\sg$ that minimizes the makespan.

Vector scheduling (a.k.a. vector load balancing) was first considered by Chekuri and Khanna \cite{ChekuriK04} who gave an $O(\log^2 d)$-approximation algorithm for the problem.
Meyerson, Roytman and Tagiku in \cite{MeyersonRT13} improved the approximation guarantee to a factor $O(\log d)$; in fact, they gave an $O(\log d)$-competitive algorithm for the online version of the problem where jobs arrive one at a time and have to be assigned irrevocably to a machine on arrival.
The current best approximation for the problem is an $O(\log d/\log \log d)$-competitive algorithm by Im, Kell, Kulkarni and Panigrahi \cite{ImKKP19}.
In terms of hardness, very recently, Sai Sandeep \cite{Sandeep21} showed that, for any $\eps > 0$, (offline) vector scheduling is hard to approximate to a factor $O((\log d)^{1-\eps})$ under some complexity theoretic assumptions.

%%%% [edit-sharat] explicitly calling j as a job and r as a dimension for clarity.
{
A natural lower bound on the optimal makespan in vector scheduling is given by:
\begin{equation} \label{lbdefn}
\LB := \max \left\{\max_{\substack{\text{job } j \in J \\ \text{dimension } r \in [d]}} p_{j,r}, \, \quad \frac{1}{m} \cdot \max_{\text{dimension } r \in [d]} \sum_{\text{job } j \in J} p_{j,r} \right\}.
\end{equation}
}
\noindent {
The first term above arises because each job has to assigned to some machine, and the second term arises because in each dimension $r$, a total load of $\sum_j p_{j,r}$ is distributed across $m$ machines.}  

{Motivated by applications in stochastic load balancing,
the main question that initiated this work is: how well can we approximate vector
scheduling with respect to the natural lower bound $\LB$?}
{\em We devise a simple $\LB$-relative
randomized $O(\log d)$-approximation algorithm} (see Theorem~\ref{thm:mainvs}). 

We make a few remarks comparing this result with some of of the (previously mentioned)
prior work on vector scheduling. The $O(\log d)$ approximation guarantee of Meyerson et  
al.~\cite{MeyersonRT13} is also with respect to $\LB$, although it is not explicitly
stated in this form.  
While their guarantee is deterministic and holds also in the online setting, 
our chief notable feature is the simplicity of our algorithm and analysis.%
\footnote{The algorithm of Meyerson et al. is also fairly simple---it assigns the
  recently-arrived job to the machine that leads to the least increase in an exponential
  potential function---but the analysis is slightly involved.}
From a pure approximation-standpoint, the $O(\log d/\log \log d)$-approximation of
\cite{ImKKP19} is of course better, %strictly superior to our result 
but from the description of their algorithm it is unclear if their guarantee holds
with respect to the lower bound $\LB$; also, their algorithm and analysis are 
significantly more involved. 
Finally, we note that the hardness result in \cite{Sandeep21} essentially shows that our approximation guarantee is tight up to $\poly(\log \log d)$ factors. 

{While the simplicity and rather evident nature of the lower bound $\LB$ make it appealing
to design $\LB$-relative approximation algorithms for vector scheduling,
there are also other benefits of working with this lower bound.}
As an application, we show how our results on vector scheduling lead to useful approximation algorithms for a problem in stochastic load balancing that has recently received renewed interest. 
In the stochastic load balancing model that we consider, job sizes are (scalar) random variables with known distributions, and we have $m$ identical machines.
That is, the size of a job $j \in J$ is distributed as a nonnegative random variable $X_j$.
Throughout this paper, we assume that the job random variables are independent of each other.
An assignment $\sg: J \to [m]$ induces a random $m$-dimensional load vector $\lvecsg$, where the $i^\mathrm{th}$ coordinate of $\lvecsg$ is the (random) total processing time assigned to machine $i \in [m]$ i.e., $\lvecsg_i := \sum_{j : \sg(j) = i} X_j$.
In the \emph{stochastic minimum norm load balancing} (\stochlb) problem, we have $n$ stochastic jobs $\{X_j\}_{j \in J}$, $m$ identical machines, and a monotone symmetric norm $f : \R^m \to \Rp$.
Recall that $f$ being a norm means that: (i) $f(x) = 0$ if and only if $x = 0$; (ii) $f(\scal x) = |\scal| f(x)$ for all $x \in \R^m, \scal \in \R$; and (iii) $f(x+y) \leq f(x) + f(y)$ for all $x,y \in \R^m$.
A \emph{monotone} norm $f$ satisfies $f(x) \leq f(y)$ for all $0 \leq x \leq y$ (coordinate-wise inequality), and \emph{symmetry} is the property that permuting the coordinates of $x$ does not change its norm.
The goal in \stochlb is to find an assignment $\sg$ that minimizes $\E{f(\lvecsg)}$ --- the expected $f$-norm of the induced load vector --- where the expectation is over the randomness in the job-size distributions.
We emphasize that all jobs are assigned to machines up front without the knowledge of the job-size realizations.
We sometimes use \stochflb to explicitly indicate the norm $f$ that we are considering in an instance of \stochlb.  
 
Kleinberg, Rabani and Tardos \cite{KleinbergRT00} were the first to consider the \emph{stochastic makespan minimization} problem (\stochlb when $f$ is the $\ell_\infty$ norm) and gave an $O(1)$-approximation algorithm.
Subsequent work by \cite{GoelI99,GuptaKNS21,Molinaro19,ChakrabartyS19a,ChakrabartyS19b,IbrahimpurS20,DeKLN20,IbrahimpurS21} have focused on obtaining approximation algorithms for various \stochlb settings: (i) unrelated-machines setting where job-size distributions can depend on the machine; (ii) the norm $f$ is an $\ell_p$ norm or a $\topl$ norm (where $\Topl{x}$ is defined as the sum of $\ell$ largest coordinates of $x$); and (iii) job-sizes distributions are Bernoulli, Exponential, Poisson random variables, or even deterministic ($X_j = p_j$ with probability $1$).
We elaborate on related work in Section~\ref{related}.
The result most relevant to this work is the $O(\log m/\log \log m)$-approximation algorithm of \cite{IbrahimpurS20} for \stochlb when $f$ is an arbitrary monotone symmetric norm and jobs are arbitrarily distributed.
In this work, we give an exponential improvement in the approximation guarantee for \stochlb when machines are identical, and to do this we crucially use the $\LB$-relative approximation algorithm for vector scheduling.

\subsection{Our Results and Techniques} \label{ourwork}

We formally state our contributions in this work with a brief overview of the techniques used.
Our first main result is a simple $O(\log d)$-approximation algorithm for vector scheduling where the approximation guarantee is with respect to the natural lower bound.

\begin{theorem} \label{thm:mainvs}
There is a randomized algorithm for the $d$-dimensional vector scheduling problem that computes an assignment $\sg$ with makespan at most $O(\log d) \cdot \LB$.
The algorithm runs in time polynomial in the input size, and succeeds in finding an approximate assignment with probability at least $2/3$.
\end{theorem}

Our vector-scheduling algorithm in Theorem~\ref{thm:mainvs} is extremely simple to implement since it is based on randomly sampling jobs to be processed on a machine. Its analysis is equally simple because it only uses the well-known Chernoff tail bounds.
By scaling job-size vectors, we may assume without loss of generality that $\LB = 1$. 
In particular, this implies $p_{j,r} \in [0,1]$ for any job $j \in J$ and dimension $r \in [d]$, and $\sum_j p_{j,r} \leq m$ for any $r$. 
For simplicity, further assume that $\sum_j p_{j,r} = m$ for all $r$.
Consider a random job-set $S \subseteq J$ where job $j$ is independently included in $S$ with probability $c\log d/m$ for a sufficiently large constant $c \geq 1$.
By our choice of $S$, for any dimension $r \in [d]$, the expected total size of jobs in $S$ in that dimension is $c \log d$.
By Chernoff tail bounds, in any dimension $r$, we have $\pr{1 \leq \sum_{j \in S} p_{j,r} \leq O(\log d)} \leq 1/\poly(d)$.
By union bound, with some positive probability, $\sum_{j \in S} p_{j,r}$ is in $[1,O(\log d)]$ simultaneously for all dimensions $r \in [d]$.
We assign the job-set $S$ to one of the $m$ machines, and recurse on the residual instance with job-set $J \setminus S$ and $m-1$ identical machines; note that $\LB$ for the residual instance is at most $1$.
We formalize the argument in Section~\ref{sec:vs}.

Our second main result is an $O(\log \log m)$-approximation for \stochlb.

\begin{theorem} \label{thm:stochlbapx}
We can compute (in polynomial time) an assignment $\sg : J \to [m]$ such that, with probability at least $2/3$, $\sg$ is an $O(\log \log m)$-approximation to the given instance of \stochlb.
\end{theorem}

In fact, we prove a stronger result: we return an assignment that \emph{simultaneously} achieves an $O(\log \log m)$-approximation for \emph{all} monotone, symmetric norms.

\begin{theorem} \label{thm:simulapx}
Consider $n$ stochastic jobs $\{X_j\}_{j \in J}$ and $m$ identical machines.
We can compute (in polynomial time) an assignment $\sg : J \to [m]$ such that, with
probability at least $2/3$, $\sg$ is simultaneously an $O(\log \log m)$-approximation to all \stochflb[\msnorm] instances where $\msnorm: \R^m \to \Rp$ is a monotone, symmetric norm.
Moreover, given an $\LB$-relative $\gm(d)$-approximation algorithm $\A$ for $d$-dimensional vector scheduling, we can use one call of $\A$ to obtain an $O(\gm(\log m))$ simultaneous-approximation for \stochflb[\msnorm] for all monotone, symmetric norms $\msnorm$.
\end{theorem}

The randomization above stems from our randomized guarantee for vector scheduling. 
We remark that the probability of success in Theorem~\ref{thm:simulapx} is of a detectable (polytime-verifiable) event (see Remark~\ref{badevent}). 
Therefore, we can lower the failure probability to $\dl$, for any $\dl > 0$, by repeating the algorithm in Theorem~\ref{thm:simulapx} $O(\log (1/\dl))$ times. 

At a high level, to prove Theorem~\ref{thm:simulapx} we utilize two key ideas from \cite{IbrahimpurS20}.
First, a technical result from \cite{IbrahimpurS20} called as \emph{approximate stochastic majorization} (Theorem~\ref{thm:stochmaj}) gives a reduction from the problem of approximating all monotone symmetric norms to that of approximating only a small collection of $\topl$ norms.
Formally, a simultaneous $\al$-approximation to all \stochflb[\topl] instances with $\ell \in \{1,2,4,\dots,2^{\floor{\log_2 m}}\}$ implies a simultaneous $O(\al)$-approximation for all \stochlb instances (see Theorem~\ref{thm:simultoplapx}).
Second, we specialize the $O(1)$-approximation algorithm for \stochflb[\topl] found in \cite{IbrahimpurS20} to the identical-machines setting.
We first note that, modulo some technical assumptions, \stochflb[\topl] problem on identical machines is essentially an instance of $1$-dimensional vector scheduling where the (scalar) size of a job $j$ is a certain sophisticated function of its job-size distribution $X_j$.
More crucially, we show that an $\LB$-relative $\al$-approximate assignment for this $1$-dimensional vector scheduling instance yields an $O(\al)$-approximate assignment for the \stochflb[\topl] instance.
Since the expression of $\LB$ for $d$-dimensional vector scheduling has a certain independence among its $d$ dimensions, we can approximate the simultaneous \stochflb[\topl] problem (for all powers-of-$2$ $\ell$) via the $O(\log m)$-dimensional vector scheduling problem.
As Theorem~\ref{thm:logdapx} gives us an $O(\log d)$-approximation for $d$-dimensional vector scheduling, we get an $O(\log \log m)$-approximation in Theorems~\ref{thm:stochlbapx}~and~\ref{thm:simulapx}.
We remark that our algorithms for stochastic load balancing applications only require black-box access to $\LB$-relative approximation algorithms for vector scheduling. 
If the vector-scheduling algorithm is deterministic (randomized), then our algorithm for Theorem~\ref{thm:simulapx} is also deterministic (randomized). 

\subsection{Related Work} \label{related}

The $d$-dimensional vector scheduling problem was first considered by Chekuri and Khanna in~\cite{ChekuriK04} where they gave an $O(\log^2 d)$-approximation algorithm and showed that the problem is NP-hard to approximate within any constant factor.
They also gave a PTAS when $d$ is a fixed constant.
Meyerson et al. \cite{MeyersonRT13} gave an improved $O(\log d)$-approximation and the current best approximation factor of $O(\log d/\log \log d)$ is due to Harris and Srinivasan \cite{HarrisS19} and Im et al. \cite{ImKKP19}.
The results in \cite{MeyersonRT13,ImKKP19} also hold in the online setting, and the result of \cite{HarrisS19} works even in the unrelated-machines setting where the size vector of a job $j$ can depend on the machine that it is processed on.
Recently, Sai Sandeep \cite{Sandeep21} gave very strong inapproximability results for (offline) vector scheduling indicating that the current best results are almost optimal: assuming $\mfont{NP} \not \subseteq \mfont{ZPTIME}\bigl( n^{(\log n)^{O(1)}} \bigr)$, they rule out an $O((\log d)^{1-\eps})$-approximation for any $\eps > 0$, and under the weaker assumption of $\mfont{P} \neq \mfont{NP}$, they rule out a $\poly(\log \log d)$-approximation.
We also mention that factor $O(\log d/\log \log d)$ is the best possible competitive ratio for any (deterministic or randomized) online vector scheduling algorithm (see \cite{ImKKP19,AzarCP18}). 

Growing interest in \emph{optimization under uncertainty} has led to some recent work on stochastic load balancing under various norm objectives.
As mentioned before, Kleinberg et al. \cite{KleinbergRT00} were the first to investigate stochastic load balancing on identical machines for the makespan objective (i.e., \stochflb[$\ell_\infty$]) and gave an $O(1)$-approximation algorithm for the problem. 
This result was generalized to the unrelated-machines setting by Gupta, Kumar, Nagarajan and Shen \cite{GuptaKNS18,GuptaKNS21}, and subsequently, Molinaro \cite{Molinaro19} generalized the result to all $\ell_p$ norms.
Ibrahimpur and Swamy \cite{IbrahimpurS20,IbrahimpurS20b} were the first to consider \stochlb in its most general form (unrelated machines, arbitrary monotone symmetric norm objectives, and arbitrary job-size distributions) and gave the following approximation guarantees: (i) $O(1)$-approximation when jobs are weighted Bernoulli random variables; (ii) $O(1)$-approximation when $f$ is a $\topl$ norm; and (iii) $O(\log m/\log \log m)$-approximation for the most general setting.
With an eye towards obtaining small approximation factors, Goel and Indyk \cite{GoelI99} considered stochastic makespan minimization (on identical machines) when job-sizes follow a structured distribution. 
Among other results, they obtained a simple $2$-approximation when jobs are Poisson-distributed.
Very recently, there have been two works on improving the approximation factor for stochastic min-norm load balancing with Poisson jobs (\poislb). 
First, De, Khanna, Li and Nikpey \cite{DeKLN20} gave a PTAS for the makespan objective when machines are identical.
Second, Ibrahimpur and Swamy \cite{IbrahimpurS21} gave a $(2+\eps)$-approximation for the most general version of \poislb (arbitrary monotone symmetric norm objectives and unrelated-machines setting), and a PTAS when machines are identical.

\section{Vector Scheduling} \label{sec:vs}

Recall that in $d$-dimensional vector scheduling, we have a set $J$ of vector jobs and $m$ identical machines.
The size-vector of job $j$ is $p_j \in \Rp^d$.
To keep the notation simple, we reserve $i \in [m]$ to index the machine-set, $j \in J$ to index the job-set, and $r \in [d]$ to index the dimensions.
We seek an assignment $\sg : J \to [m]$ that minimizes the makespan $\max_{i,r} \bigl( \sum_{j : \sg(j) = i} p_{j,r} \bigr)$.
The natural lower bound $\LB$ on the optimal makespan is $\max\{\max_{j,r} p_{j,r}, \max_r \sum_j p_{j,r}/m \}$ (see \eqref{lbdefn}).
By scaling, we may assume without loss of generality that $\LB = 1$; this implies 
$p_{j,r} \in [0,1]$ for all $j,r$, and $\sum_j p_{j,r} \leq m$ for all $r$.

We prove our main result on vector scheduling (Theorem~\ref{thm:mainvs}) by proving a slightly stronger result.

\begin{theorem} \label{thm:logdapx}
Consider an instance of vector scheduling with $p_{j,r} \in [0,1]$ for all $j \in J,r \in [d]$, and $\sum_{j \in J} p_{j,r} \leq m \log d$ for all $r \in [d]$.
There is a randomized algorithm that produces an assignment $\sg$ whose makespan is $O(\log d)$.
The algorithm succeeds with probability at least $2/3$.
\end{theorem}

The only tool that we use in the proof of Theorem~\ref{thm:logdapx} is Chernoff tail bounds, which we state below.

\begin{lemma}[Chernoff Bounds] \label{chernoff}
Let $\{\scrv_j\}$ be a finite collection of independent $[0,1]$-bounded random variables, and let $\sumrv = \sum_j \scrv_j$ denote their sum.
Then,
\vspace{-1em}
\begin{enumerate}[(a)]
\item for any $\mu \geq \E{S}$ and $\delta \geq 0$, we have $\pr{S \geq (1+\delta) \mu} \leq \exp\bigl(-\frac{\delta^2 \mu}{2+\delta}\bigr)$. 
\item for any $\mu \leq \E{S}$ and $\delta \in [0,1]$, we have $\pr{S \leq (1-\delta)\mu} \leq \exp\bigl(-\frac{\delta^2 \mu}{2}\bigr)$.
\end{enumerate}
\end{lemma}

As mentioned in Section~\ref{ourwork}, our algorithm for vector scheduling is based on finding a subset of jobs $S \subseteq J$ that can all be processed on a single machine by incurring a load of at most $O(\log d)$ in each dimension $r \in [d]$,  while ensuring that the residual problem is also a valid sub-instance of the problem with $m-1$ machines.
We formalize the argument below.

\begin{lemma} \label{lem:goodsubset}
Suppose $m \geq 7$ and $d \geq 2$. 
Consider a random job-set $S \subseteq J$ where job $j \in J$ is independently included in $S$ with probability $q := 7/m$.
With probability at least $1/2$, for all $r \in [d]$ we have $\sum_{j \in S} p_{j,r} \leq 14 \log d$ and $\sum_{j \in J \setminus S} p_{j,r} \leq (m-1) \log d$.
\end{lemma}
\begin{proof}
The proof is based on a straightforward application of Chernoff bounds. 
By our choice of $q$, for any $r$ we have $\E{\sum_{j \in S} p_{j,r}} \leq 7 \log d$.
Using Lemma~\ref{chernoff}(a) with $\delta = 1$, for any $r$ we have:
\[
\probb{ \sum_{j \in S} p_{j,r} \geq 14 \log d } \leq \exp\bigl(\frac{-7\log d}{3}\bigr) \leq \frac{1}{d^2}. 
\]

Next, we use Chernoff bounds for the lower tail to prove the remaining size-bound on jobs in $J \setminus S$.
Let $B := \{ r \in [d] : \sum_{j \in J} p_{j,r} > (m-1) \log d \}$.
Since $m \geq 7$, for any $r \in B$ we have $\E{\sum_{j \in S} p_{j,r}} \geq (1-1/m) \cdot 7 \log d \geq 6 \log d$.
Using Lemma~\ref{chernoff}(b) with $\delta = 5/6$ we get:
\[
\pr{ \sum_{j \in S} p_{j,r} \leq \log d } \leq \exp\bigl(-\frac{25}{36}\cdot\frac{6\log d}{2}\bigr) \leq \frac{1}{d^2}.
\]

By union-bound, with probability at least $1/2$, we have $\sum_{j \in S} p_{j,r} \leq 14 \log d$ for all $r \in [d]$, and $\sum_{j \in S} p_{j,r} \geq \log d$ for all $r \in B$.
The desired bounds follow.
\end{proof}

\begin{proofof}{Theorem~\ref{thm:logdapx}}
We call an instance of vector scheduling as a \emph{valid instance with $m$ machines} if $p_{j,r} \in [0,1]$ for all $j,r$, and $\sum_j p_{j,r} \leq m \log d$ for all $r$.
Note that in Theorem~\ref{thm:logdapx} we are given a valid instance with $m$ machines.
We will assume that $d \geq 2$, since otherwise a greedy list-scheduling algorithm gives
an easy $2$-approximation. 

Let $N := \ceil{\log_2 (3m)} = O(\log m)$.
We now describe a randomized algorithm for vector scheduling.
The overall procedure has at most $m$ iterations, and if the algorithm is successful, it produces an assignment with makespan $O(\log d)$.
Consider the start of an iteration $t$ for some $t \geq 1$. 
We have a valid instance with $m-t+1$ identical machines.
If $m-t+1 = 6$, then we simply assign all remaining jobs to a single machine; note that $\sum_j p_{j,r} \leq 6 \log d$ for all dimensions $r$, so we are not assigning too much load to this machine in any dimension.
Otherwise, $m-t+1 \geq 7$.
Consider random subsets $S_1,\dots,S_N$ where each $S_\ell \subseteq J, \ell \in [N]$ is obtained by independently including job $j$ in $S_\ell$ with probability $q := 7/m$.
By Lemma~\ref{lem:goodsubset}, with probability at most $2^{-N}\leq\frac{1}{3m}$,
none of the $S_\ell$ sets satisfy the conclusion of the lemma; in this case we terminate
the algorithm with a failure. 
Otherwise, for some $\ell \in [N]$, we have for all $r$, $\sum_{j \in S_\ell} p_{j,r} \leq 14 \log d$ and $\sum_{j \in J \setminus S_\ell} p_{j,r} \leq (m-t) \log d$.
In this case, we assign all jobs in $S_\ell$ to one machine and end the iteration with a valid instance with $m-t$ machines and residual job-set $J \setminus S_\ell$.
The probability that the algorithm terminates with a failure in any given iteration is at most $1/3m$, so the overall failure probability of the algorithm is at most $1/3$.
In other words, with probability at least $2/3$, we obtain an assignment whose makespan is
at most $14 \log d$.
\end{proofof}

We remark that if we assume $\sum_j p_{j,r} \leq m\cdot U$ for all dimensions $r$ in the
statement of Theorem~\ref{thm:logdapx} (where $U$ does not depend on $m$), then we can
easily adapt the above arguments to obtain makespan $O(\max(\log d,U))$.
Suppose $U>\log d$. 
Sampling jobs as before with probability $7/m$ now ensures that the random subset
$S$ satisfies 
$\probb{\sum_{j \in S} p_{j,r} \leq 14 U,\ \sum_{j\in J\setminus S}p_{j,r}\leq(m-1)U}\geq 1-\frac{2}{d^2}$, 
for $m\geq 7$. 
Given this, we proceed exactly as in Theorem~\ref{thm:logdapx}, where a valid instance
with $m$ machines now means that the total load in each dimension $r$ is at most $m\cdot U$.

\section{\boldmath Stochastic Minimum-Norm Load Balancing} \label{sec:stochlb}

We now utilize our results on vector scheduling to derive approximation algorithms for {\em stochastic minimum norm load balancing} ($\stochlb$), crucially leveraging the fact that our approximation guarantee for vector scheduling is with respect to the natural lower bound $\LB$ (see \eqref{lbdefn}).

Recall that in \stochlb we have a set $J$ of $n$ independent stochastic jobs and $m$ identical machines.
The processing time of a job $j$ is a nonnegative random variable $X_j$ whose distribution is given to us. (We will only however need a certain kind of access to these distributions.) 
We are also given a monotone symmetric norm $f : \R^m \to \Rp$.
Our goal in \stochlb is to find an assignment $\sg$ that minimizes $\E{f(\lvecsg)}$, where $\lvecsg_i := \sum_{j : \sg(j) = i} X_j$ for $i \in [m]$.
Also recall that we use \stochflb[\msnorm] to explicitly indicate the norm $\msnorm$ that we are considering in an instance of \stochlb.

In this section we prove Theorem~\ref{thm:simulapx} by giving a simultaneous $O(\log \log m)$-approximation for all \stochflb[\msnorm] instances where $\msnorm$ is a monotone, symmetric norm.
Theorem~\ref{thm:stochlbapx} immediately follows from Theorem~\ref{thm:simulapx}.
The proof of Theorem~\ref{thm:simulapx} is based on a reduction to vector
scheduling with $O(\log m)$-dimensional vectors, by utilizing suitable tools from prior work~\cite{KleinbergRT00,IbrahimpurS20}. 
We first discuss these relevant tools, and show in Section~\ref{subsec:reduction} how they lead to the proof of Theorem~\ref{thm:simulapx}.
We first introduce some notation.
For a nonnegative random variable $Z$ and a scalar $\scal$, we define the \emph{truncated random variable} $\trv{Z}{\scal} := Z \cdot \boldone_{Z < \scal}$ as the random variable that takes size $Z$ when the event $\{Z < \scal\}$ happens, and size $0$ otherwise.
We also define the \emph{exceptional random variable} $\erv{Z}{\scal} := Z \cdot \boldone_{Z \geq \scal}$ as the random variable that takes size $Z$ when the event $\{Z \geq \scal\}$ happens, and size $0$ otherwise.
Note that $Z = \trv{Z}{\scal} + \erv{Z}{\scal}$.
We reserve $\msnorm : \R^m \to \Rp$ to denote an arbitrary monotone symmetric norm and $Y = (Y_1,\dots,Y_m)$ to denote an arbitrary $m$-dimensional random vector whose coordinates $\{Y_i\}_{i \in [m]}$ form an independent collection of nonnegative random variables.
For brevity, we will say that $Y$ follows a product distribution on $\Rp^m$.
Observe that load vectors $\lvecsg$ that arise in \stochlb instances with $m$ machines follow a product distribution on $\Rp^m$.

Let $\POS=\POS(m):= \{1,2,4,\dots,2^{\floor{\log_2 m}}\}$. 
The first tool that we utilize shows importantly that we can control $\E{\msnorm(Y)}$ by controlling a collection of simpler norms called $\topl$ norms (Theorem~\ref{thm:stochmaj}).

\begin{definition}[$\topl$ norm] \label{def:topl}
For any $\ell \in [m]$, the {\em $\topl$ norm} is defined as follows: for $x \in \Rp^m$, $\Topl{x}$ is the sum of the $\ell$ largest coordinates of $x$, i.e., $\Topl{x} = \sum_{i=1}^{\ell} x^\da_i$, where $x^\da$ denotes the vector $x$ with its coordinates sorted in non-increasing order.
\end{definition}

Like $\ell_p$ norms, $\topl$ norms give us a way to interpolate between the $\mathrm{max}$-norm ($\topl[1]$) and the $\mathrm{sum}$-norm ($\topl[m]$).
A well-known mathematical result called the \emph{majorization inequality}~\cite{InequalitiesHLP} (see also~\cite{ChakrabartyS19a}) shows that controlling all $\topl$ norms gives us a handle on all monotone symmetric norms: for $x,y \in \Rp^m$, if $\Topl{x} \leq \Topl{y}$ for all $\ell \in [m]$, then $\msnorm(x) \leq \msnorm(y)$ for every monotone, symmetric norm $\msnorm$.
Recently, in~\cite{IbrahimpurS20}, we proved a substantial stochastic generalization
(qualitatively speaking) of this result.

\begin{theorem}[Approximate Stochastic Majorization, Theorem 5.2 in~\cite{IbrahimpurS20b}] \label{thm:stochmaj}
Let $Y = (Y_1,\dots,Y_m)$ and $W = (W_1,\dots,W_m)$ follow product distributions on $\Rp^m$, and let $\al$ be a positive scalar.
\begin{enumerate}[(i)]
\item If $\E{\Topl{Y}} \leq \al \E{\Topl{W}}$ for all $\ell \in [m]$, then $\E{\msnorm(Y)} \leq 28 \cdot \al \E{\msnorm(W)}$.
\item If $\E{\Topl{Y}} \leq \al \E{\Topl{W}}$ for all $\ell \in \POS(m)$, then \newline ${\E{\msnorm(Y)} \leq 56 \cdot \al \E{\msnorm(W)}}$.
\end{enumerate} 
\end{theorem}

By Theorem~\ref{thm:stochmaj}(ii), the task of proving Theorem~\ref{thm:simulapx} reduces to proving the following.

\begin{theorem} \label{thm:simultoplapx}
Consider $n$ stochastic jobs $\{X_j\}_{j \in J}$ and $m$ identical machines.
With probability at least $2/3$, we can obtain an assignment $\sg : J \to [m]$ that is simultaneously an $O(\log \log m)$-approximation to all \stochflb[\topl] instances with $\ell \in \POS$.
\end{theorem}

We prove Theorem~\ref{thm:simultoplapx} using an approximation algorithm for
$|\POS|$-dimensional vector scheduling {\em with respect to the natural lower bound $\LB$}. 
Since $|\POS| = O(\log m)$, and Section~\ref{sec:vs} presents such an 
$O(\log d)$-approximation algorithm for $d$-dimensional vector scheduling, we obtain the desired result. 

\begin{subsection}{\boldmath Handling a single $\topl$ norm} \label{subsec:exptopl}

We next discuss how to tackle stochastic load balancing with a single $\topl$ norm. 
Combining the ingredients here for all $\ell \in \POS$ will yield our reduction to $|\POS|$-dimensional vector scheduling (see Section~\ref{subsec:reduction}).

We fix $\ell \in\POS$ throughout this subsection.
In~\cite{IbrahimpurS20}, the authors give an $O(1)$-approximation for \stochflb[\topl], even in the more-general setting of unrelated machines via an LP-based algorithm, but for identical machines, the underlying arguments can be substantially simplified.
In particular, \cite{IbrahimpurS20} show that, $\E{\Topl{Y}}$ can be approximated via a separable expression (in the $Y_i$s) when $Y$ follows a product distribution
(Theorem~\ref{thm:exptopl}), and show how one can deal with this proxy expression by
generalizing certain arguments from~\cite{KleinbergRT00}. 
Mimicking the proof strategy in~\cite{KleinbergRT00}, which gives an $O(1)$-approximation for \stochflb[{\topl[1]}], then shows that \stochflb[\topl] can
be reduced to (scalar) makespan minimization (i.e., $1$-dimensional vector scheduling).
The following technical theorem gives a handle on $\E{\Topl{Y}}$ when $Y$ follows a product distribution on $\Rp^m$.

\begin{theorem} 
\label{thm:exptopl}
Let $\scal$ be a positive scalar.
\begin{enumerate}[(i)]
\item \textnormal{(Lemma 4.1 in~\cite{IbrahimpurS20b})}\ \ 
If $\sum_{i \in [m]} \E{\erv{Y_i}{\scal}} \leq \ell \scal$, then $\E{\Topl{Y}} \leq 2 \ell \scal$.
\item \textnormal{(Lemma 4.2 in~\cite{IbrahimpurS20b})}\ \ 
If $\sum_{i \in [m]} \E{\erv{Y_i}{\scal}} > \ell \scal$, then $\E{\Topl{Y}} > \ell \scal/2$.
\end{enumerate}
\end{theorem}

The upshot of Theorem~\ref{thm:exptopl} is that if $\scal$ is such that $\sum_{i \in [m]} \E{\erv{Y_i}{\scal}}$ is roughly $\ell \scal$, then the {\em separable expression} $\sum_{i \in [m]} \E{\erv{Y_i}{\scal}}$ is a constant-factor proxy for 
$\E{\Topl{Y}}$.
In \stochlb, we have $Y_i=\lvecsg_i$ (where $\sg:J\mapsto[m]$ is the job-assignment),
which is a sum of independent random variables, and we next need to address how to obtain
a handle on $\E{\erv{Y_i}{\scal}}$. 
Similar to \cite{KleinbergRT00,GuptaKNS21,IbrahimpurS20}, given the ``scale'' $\scal$, we
can decompose any instance of \stochflb[\topl] into an ``easy'' sub-instance, for which
it is trivial to bound $\E{\erv{Y_i}{\scal}}$ and any assignment yields an
$O(1)$-approximation, and a (more) ``difficult'' sub-instance, where we need some 
technical results to control $\E{\erv{Y_i}{\scal}}$ and obtain an $O(1)$-approximation.
Lemmas~\ref{lem:exceptopl} and~\ref{lem:trunctopl} indicate how one may define these two subinstances.  

\begin{lemma} \label{lem:exceptopl}
Let $\scal$ be a positive scalar.
Consider an instance of \stochflb[\topl] where for each job $j$, the distribution of $X_j$ is supported on $\{0\} \cup [\scal,\infty)$ i.e., $\pr{0 < X_j < \scal} = 0$.
\begin{enumerate}[(i)]
\item If $\sum_{j \in J} \E{X_j} \leq \ell \scal$, then for any assignment $\sg$, we have $\E{\Topl{\lvecsg}} \leq 2 \ell \scal$.
\item If $\sum_{j \in J} \E{X_j} > \ell \scal$, then for any assignment $\sg$, we have $\E{\Topl{\lvecsg}} > \ell \scal/2$.
\end{enumerate}
\end{lemma}
\begin{proof}
Fix an assignment $\sg$ and let $Y := \lvecsg$ denote the induced random load vector.
By Theorem~\ref{thm:exptopl}, it suffices to argue that $\sum_{i \in [m]} \E{\erv{Y_i}{\scal}} = \sum_{j \in J} \E{X_j}$.
Since the distribution of each $X_j$ is supported on $\{0\} \cup [\scal,\infty)$, the distribution of the $i^\mathrm{th}$ machine's load $Y_i = \sum_{j \in J : \sg(j) = i} X_j$ is also supported on $\{0\} \cup [\scal,\infty)$.
Thus, 
\[
\sum_{i \in [m]} \E{\erv{Y_i}{\scal}} = \sum_{i \in [m]} \E{Y_i} = \sum_{i \in [m]} \sum_{ j : \sg(j) = i} \E{X_j} = \sum_{j \in J} \E{X_j},
\] 
and so both parts follow.
\end{proof}

Thus, the difficulty in bounding $\E{\erv{Y_i}{\scal}}$ (where
$Y_i=\lvecsg_i$) arises when the $X_j$ job random variables are supported on $[0,\scal)$. 
We need the notion of \emph{effective size} to deal with such difficult
instances. 

\begin{definition}[Effective Size] \label{def:effsize}
For a nonnegative random variable $\scrv$ and a scalar parameter $\ld > 1$, the $\ld$-effective size $\eff{\scrv}$ of $\scrv$ is defined as $\log_\ld\E{\ld^{\scrv}}$. 
We define $\eff[1]{\scrv}$ to be $\E{\scrv}$.
\end{definition}

The benefit of effective sizes follows because Lemmas~\ref{lem:smallbeta}
and~\ref{lem:largebeta} (from~\cite{IbrahimpurS20}) show that, for
independent random variables $\{\scrv_j\}$, we can control
$\E{\erv{\bigl(\sum_j\trv{\scrv_j}{\scal}\bigr)}{\scal}}$ by bounding the
effective sizes of certain random variables related to the $\trv{\scrv_j}{\scal}$ random
variables.
We leverage this to obtain the following result, which is implicit in~\cite{IbrahimpurS20}
and is a generalization of Lemma~3.4 from~\cite{KleinbergRT00};
for completeness, we present the proof in Appendix~\ref{append-stochlb}, 
%for the sake of completeness, 
since this statement does not explicitly appear in prior work. 

\begin{lemma} \label{lem:trunctopl}
Let $\scal$ be a positive scalar and let $\ld := \floor{\frac{2m}{\ell}}$.
Consider an instance of stochastic $\topl$-norm load balancing where for each job $j \in J$, the distribution of $X_j$ is supported on $[0,\scal)$ i.e., $\pr{X_j \geq \scal} = 0$.
\begin{enumerate}[(i)]
\item If $\sum_{j \in J} \eff{X_j/4\scal} \leq 8m$, then for any assignment $\sg$ we have $\E{\Topl{\lvecsg}} \leq 32 \cdot \al \cdot \ell \scal$, where $\al := \max\{1, \max_{i \in [m]} \sum_{j : \sg(j) = i} \eff{X_j/4\scal} \}$.
\item If $\sum_{j \in J} \eff{X_j/4\scal} > 8m$, then for any assignment $\sg$ we have $\E{\Topl{\lvecsg}} > \ell \scal/2$.
\end{enumerate}
\end{lemma}

Lemmas~\ref{lem:exceptopl}~and~\ref{lem:trunctopl} yield the following approach for
obtaining a constant-factor approximation for \stochflb[\topl]. 
Let $\OPT_\ell$ denote the optimal value. 
For a given scalar $\scal$, ``split'' each job $j$ as $X_j = \trv{X_j}{\scal} + \erv{X_j}{\scal}$.
This yields the {\em exceptional sub-instance} with the exceptional job-variables
$\{\erv{X_j}{\scal}\}_j$, and the {\em truncated sub-instance} with the truncated
job-variables $\{\trv{X_j}{\scal}\}_j$. 
Clearly, for any assignment $\sg$, if $Y$ and $Y'$ denote the respective load vectors in these two sub-instances, we have $\lvecsg = Y + Y'$.

Let $\ld=\floor{\frac{2m}{\ell}}$.
Now, we can use binary search to find the \emph{right} threshold $t_\ell$, such that for
$\scal=t_\ell$, the conditions in part (i) of Lemmas~\ref{lem:exceptopl}
and~\ref{lem:trunctopl} hold, and for some $\scal \geq t_\ell/(1+\eps)$, the opposite is
true. 
This is because Lemmas~\ref{lem:exceptopl} and~\ref{lem:trunctopl} imply that for $\scal$
large enough, the conditions in part (i) of these lemmas hold; we flesh out the 
binary-search interval at the end of Section~\ref{subsec:reduction}.
So we have:  
(1) $\sum_{j \in J} \E{\erv{X_j}{t_\ell}} \leq \ell t_\ell$ and 
$\sum_{j \in J} \eff{\trv{X_j}{t_\ell}/4t_\ell} \leq 8m$; and
(2) for some $\scal\geq t_\ell/(1+\eps)$, 
we have $\sum_{j \in J} \E{\erv{X_j}{\scal}}>\ell\scal$ or
$\sum_{j \in J} \eff{\trv{X_j}{\scal}/4\scal}>8m$. 
Due to (2), we obtain that the optimum value for the exceptional or truncated sub-instance
is $\Omega(\ell t_\ell)$, and hence $\OPT_\ell=\Omega(\ell t_\ell)$.
Due to (1), by part (i) of Lemma~\ref{lem:exceptopl}, the expected $\topl$-load of the
assignment $\sg$ in the exceptional sub-instance is $O(\ell t_\ell)$. 
Also, considering {\em deterministic} scheduling with job sizes 
$p_j:=\eff{\trv{X_j}{t_\ell}/4t_\ell}$ (which is at most $1$), we can easily compute an 
assignment $\sg$ that assigns total $p_j$-load at most $\sum_{j\in J}p_j/m+1\leq 9$ on any
machine; part (i) of Lemma~\ref{lem:trunctopl} then shows that the expected $\topl$-load
of $\sg$ in the truncated sub-instance is also $O(\ell t_\ell)$ (at most 
$32\cdot 9\cdot\ell t_\ell$).
Therefore, we have $\E{\topl(\lvecsg)} = O(\ell t_\ell)$, showing that $\sg$ is an
$O(1)$-approximate solution to the given instance of \stochflb[\topl]. 

\end{subsection}

\begin{subsection}{\boldmath Proof of Theorem~\ref{thm:simultoplapx}: simultaneous
approximation for \stochlb via vector scheduling} \label{subsec:reduction} 

We now prove Theorem~\ref{thm:simultoplapx} and, as a consequence, Theorem~\ref{thm:simulapx} as well. 
Recall that %in Theorem~\ref{thm:simultoplapx} we have a set of $n$ stochastic jobs 
%$\{X_j\}_{j \in J}$ and $m$ identical machines, and 
we seek an assignment $\sg$ that is
simultaneously a $O(\log \log m)$-approximation to \stochflb[{\topl}] for all
$\ell \in \POS = \{1,2,4,\dots,2^{\floor{\log_2 m}} \}$. 
Let $\A$ be a $\gm(d)$-approximation algorithm for the $d$-dimensional vector scheduling
problem with respect to the natural lower bound $\LB$. 
As we show below, we use $\A$ in a black-box fashion, so if $\A$ is deterministic
(respectively randomized), then our %algorithm for 
simultaneous-approximation for \stochlb is also deterministic (respectively randomized). 

The idea is simple: we essentially follow the approach used to obtain an
$O(1)$-approximation for \stochflb[\topl] (outlined above) for each $\ell\in\POS$
separately.
More precisely, we first find the right $t_\ell$, for each $\ell\in\POS$ separately, via
binary search. That is, for each $\ell\in\POS$, letting
$\ld_\ell=\floor{\frac{2m}{\ell}}$,  
we have that: (i) $\sum_{j \in J} \E{\erv{X_j}{t_\ell}} \leq \ell t_\ell$,
(ii) $\sum_{j \in J} \eff[\ld_\ell]{\trv{X_j}{t_\ell}/4t_\ell} \leq 8m$, and (iii)
$\OPT_\ell=\Omega(\ell t_\ell)$.  
We defer the details of this binary search to the end of the proof.
If we consider the exceptional and truncated sub-instances for an index $\ell\in\POS$,
then, as before, the expected $\topl$-load for the exceptional sub-instance is 
$O(\ell t_\ell)$, under any assignment. So we only need to consider the truncated
sub-instances for $\ell\in\POS$ (with the truncated random variables
$\{\trv{X_j}{t_\ell}\}_j$). 

Define $p_{j,\ell}:=\eff[\ld_\ell]{\trv{X_j}{t_\ell}/4t_\ell}$
for each job $j\in J$ and index $\ell\in\POS$.
By Lemma~\ref{lem:trunctopl} (i), 
it order to obtain expected $\topl$-load of $O(\al)\cdot\OPT_\ell$ for each
$\ell\in\POS$, it suffices to find an assignment $\sg$ such that, for each $\ell\in\POS$,
the total $p_{j,\ell}$-load on each machine $i$ is at most $\al$. In other words,
considering the {\em job vectors} $p_j=(p_{j,\ell})_{\ell\in\POS}\in\Rp^\POS$ for each job
$j\in J$, we seek a solution to this %$|\POS|$-dimensional 
{\em vector-scheduling instance} with makespan $\al$. 

When we have only one index $\ell$ (i.e., \stochflb[\topl]), it is easy to
argue (e.g., using Graham's list scheduling) that $\sum_{j\in J}p_{j,\ell}=O(m)$
and $p_{j,\ell}\leq 1$ for all $j,\ell$, implies that we can assign the jobs so that the
total $p_{j,\ell}$-load on each machine is $O(1)$. 
In our vector-scheduling instance, we have a bound on the total-load $p_{j,\ell}$ load on
each coordinate $\ell$ (due to (ii)), and we know that $p_{j,\ell}\leq 1$ for all $j,\ell$.
Thus, the natural lower bound $\LB$ for this vector-scheduling instance is at most $8$, and what we now need is an {\em $\LB$-relative approximation algorithm for vector scheduling}.
We can therefore utilize $\A$ to obtain makespan $\gm(\log m)$ for our vector-scheduling
instance, which thus yields an $O(\gm(\log m))$ approximation to \stochflb[\topl] for all
$\ell\in\POS$. Theorem~\ref{thm:logdapx} yields $\gm=O(\log d)$, which yields a
simultaneous $O(\log\log m)$-approximation.

\medskip
Finally, we furnish the details of the binary-search procedure to find the $t_\ell$s.
Fix an index $\ell\in\POS$.
We first establish suitable upper and lower bounds for $t_\ell$.
Define $\kp := \max_{j \in J} \E{X_j}$ (which can be easily computed from input data). 
Let $\vecrv$ denote the random load vector induced by an optimal solution and
$\OPT_\ell:=\E{\topl(\vecrv)}$ denote the optimal value.
By the contrapositive of part (ii) of Lemmas~\ref{lem:exceptopl}~and~\ref{lem:trunctopl},
for any $\scal \geq 2\OPT_\ell/\ell$ the following inequalities hold: $\sum_{j \in J}
\E{\erv{X_j}{\scal}} \leq \ell \scal$ and $\sum_{j \in J}
\eff[\floor{\frac{2m}{\ell}}]{\trv{X_j}{\scal}/4\scal} \leq 8m$. 
In particular, since $\OPT_\ell \leq \E{\sum_{i \in [m]} W_i} = \sum_{j \in J} \E{X_j} \leq n \kp$,
these two inequalities hold for $\scal=\hi:=2n\kp \geq 2 \OPT_\ell/\ell$.
Next, observe that for any job $k$ and $\scal = \E{X_k}/(m+2)$, we have
\[
\sum_j \E{\erv{X_j}{\scal}} \geq \E{\erv{X_k}{\scal}} \geq \E{X_k} - \scal = (m+1) \scal > \ell \scal.
\]
Thus, for any $\scal \leq \low := \kp/(m+2)$, we have $\sum_j \E{\erv{X_j}{\scal}} > \ell \scal$.

Let $\eps > 0$ be a small constant (say, $1/1000$).
Then, by doing a binary search in the interval $[\low,\hi]$,
we can find scalars $t_\ell$ and $t'_\ell$ satisfying:
\begin{enumerate}[(i)]
\item $t'_\ell < t_\ell \leq t'_\ell(1+\eps)$.
\item $\sum_{j \in J} \E{\erv{X_j}{t_\ell}} \leq \ell t_\ell$ and 
$\sum_{j \in J} \eff[\ld_\ell]{\trv{X_j}{t_\ell}/4t_\ell} \leq 8m$.
\item $\sum_{j \in J} \E{\erv{X_j}{t'_\ell}}>\ell t'_\ell$ or
$\sum_{j \in J} \eff[\ld_\ell]{\trv{X_j}{t'_\ell}/4t'_\ell}>8m$.
\end{enumerate}
Properties (i) and (iii) imply that $\OPT_\ell>\ell t'_\ell/2\geq\ell t_\ell/2(1+\eps)$. 
\qquad\qedsymbol

\begin{remark} \label{badevent}
If the algorithm $\A$ is randomized (as in Theorem~\ref{thm:logdapx}), then we detect if the assignment $\sg$ returned by solving vector scheduling satisfies $\sum_{j : \sg(j) = i} p_{j,\ell} \leq O(\log |\POS|)$ for all $\ell \in \POS$ and $i \in [m]$, i.e., is the approximation guaranteed by $\A$, and if not, return failure. Then, the success probability in the statement of Theorem~\ref{thm:simultoplapx} is the probability of success of $\A$ (which lower bounds the probability of obtaining a simultaneous $O(\log\log m)$-approximation). 
\end{remark}

\end{subsection}

\bibliographystyle{plainurl}
\bibliography{sosa_sub_arxiv}

\begin{thebibliography}{10}

\bibitem{AzarCP18}
Yossi Azar, Ilan~Reuven Cohen, and Debmalya Panigrahi.
\newblock {R}andomized {A}lgorithms for {O}nline {V}ector {L}oad {B}alancing.
\newblock In {\em Proceedings of the 29th {ACM-SIAM} Symposium on Discrete
  Algorithms}, pages 980--991, 2018.
\newblock \href {http://dx.doi.org/10.1137/1.9781611975031.63}
  {\path{doi:10.1137/1.9781611975031.63}}.

\bibitem{ChakrabartyS19a}
Deeparnab Chakrabarty and Chaitanya Swamy.
\newblock {A}pproximation {A}lgorithms for {M}inimum {N}orm and {O}rdered
  {O}ptimization {P}roblems.
\newblock In {\em Proceedings of the 51st {ACM} {SIGACT} Symposium on Theory of
  Computing}, pages 126--137, 2019.
\newblock \href {http://dx.doi.org/10.1145/3313276.3316322}
  {\path{doi:10.1145/3313276.3316322}}.

\bibitem{ChakrabartyS19b}
Deeparnab Chakrabarty and Chaitanya Swamy.
\newblock {S}impler and {B}etter {A}lgorithms for {M}inimum-{N}orm {L}oad
  {B}alancing.
\newblock In {\em Proceedings of the 27th European Symposium on Algorithms},
  pages 27:1--27:12, 2019.
\newblock \href {http://dx.doi.org/10.4230/LIPIcs.ESA.2019.27}
  {\path{doi:10.4230/LIPIcs.ESA.2019.27}}.

\bibitem{ChekuriK04}
Chandra Chekuri and Sanjeev Khanna.
\newblock {O}n {M}ultidimensional {P}acking {P}roblems.
\newblock {\em SIAM Journal on Computing}, 33(4):837--851, 2004.
\newblock \href {http://dx.doi.org/10.1137/S0097539799356265}
  {\path{doi:10.1137/S0097539799356265}}.

\bibitem{DeKLN20}
Anindya De, Sanjeev Khanna, Huan Li, and Hesam Nikpey.
\newblock {A}n {E}fficient {PTAS} for {S}tochastic {L}oad {B}alancing with
  {P}oisson {J}obs.
\newblock In {\em Proceedings of the 47th International Colloquium on Automata,
  Languages, and Programming}, pages 37:1--37:18, 2020.
\newblock \href {http://dx.doi.org/10.4230/LIPIcs.ICALP.2020.37}
  {\path{doi:10.4230/LIPIcs.ICALP.2020.37}}.

\bibitem{GoelI99}
Ashish Goel and Piotr Indyk.
\newblock {S}tochastic {L}oad {B}alancing and {R}elated {P}roblems.
\newblock In {\em Proceedings of the 40th Foundations of Computer Science},
  pages 579--586, 1999.
\newblock \href {http://dx.doi.org/10.1109/SFFCS.1999.814632}
  {\path{doi:10.1109/SFFCS.1999.814632}}.

\bibitem{GuptaKNS18}
Anupam Gupta, Amit Kumar, Viswanath Nagarajan, and Xiangkun Shen.
\newblock {S}tochastic {L}oad {B}alancing on {U}nrelated {M}achines.
\newblock In {\em Proceedings of the 29th ACM-SIAM Symposium on Discrete
  Algorithms}, pages 1274--1285, 2018.
\newblock \href {http://dx.doi.org/10.1137/1.9781611975031.83}
  {\path{doi:10.1137/1.9781611975031.83}}.

\bibitem{GuptaKNS21}
Anupam Gupta, Amit Kumar, Viswanath Nagarajan, and Xiangkun Shen.
\newblock {S}tochastic {L}oad {B}alancing on {U}nrelated {M}achines.
\newblock {\em Mathematics of Operations Research}, 46(1):115--133, 2021.
\newblock \href {http://dx.doi.org/10.1287/moor.2019.1049}
  {\path{doi:10.1287/moor.2019.1049}}.

\bibitem{InequalitiesHLP}
G.H. Hardy, J.E. Littlewood, and G.~P{\'o}lya.
\newblock {\em Inequalities}.
\newblock Cambridge Mathematical Library. Cambridge University Press, 1952.

\bibitem{HarrisS19}
David~G. Harris and Aravind Srinivasan.
\newblock {T}he {M}oser--{T}ardos {F}ramework with {P}artial {R}esampling.
\newblock {\em Journal of the ACM}, 66(5):1--45, 2019.
\newblock \href {http://dx.doi.org/10.1145/3342222}
  {\path{doi:10.1145/3342222}}.

\bibitem{IbrahimpurS20}
Sharat Ibrahimpur and Chaitanya Swamy.
\newblock {A}pproximation {A}lgorithms for {S}tochastic {M}inimum-{N}orm
  {C}ombinatorial {O}ptimization.
\newblock In {\em Proceedings of the 61st {IEEE} Foundations of Computer
  Science}, pages 966--977, 2020.
\newblock \href {http://dx.doi.org/10.1109/FOCS46700.2020.00094}
  {\path{doi:10.1109/FOCS46700.2020.00094}}.

\bibitem{IbrahimpurS20b}
Sharat Ibrahimpur and Chaitanya Swamy.
\newblock {A}pproximation {A}lgorithms for {S}tochastic {M}inimum {N}orm
  {C}ombinatorial {O}ptimization.
\newblock {\em CoRR}, abs/2010.05127, 2020.
\newblock URL: \url{https://arxiv.org/abs/2010.05127}.

\bibitem{IbrahimpurS21}
Sharat Ibrahimpur and Chaitanya Swamy.
\newblock {M}inimum-{N}orm {L}oad {B}alancing {I}s ({A}lmost) as {E}asy as
  {M}inimizing {M}akespan.
\newblock In {\em Proceedings of the 48th International Colloquium on Automata,
  Languages, and Programming}, pages 81:1--81:20, 2021.
\newblock \href {http://dx.doi.org/10.4230/LIPIcs.ICALP.2021.81}
  {\path{doi:10.4230/LIPIcs.ICALP.2021.81}}.

\bibitem{ImKKP19}
Sungjin Im, Nathaniel Kell, Janardhan Kulkarni, and Debmalya Panigrahi.
\newblock {T}ight {B}ounds for {O}nline {V}ector {S}cheduling.
\newblock {\em SIAM Journal on Computing}, 48(1):93--121, 2019.
\newblock \href {http://dx.doi.org/10.1137/17M111835X}
  {\path{doi:10.1137/17M111835X}}.

\bibitem{KleinbergRT00}
Jon~M. Kleinberg, Yuval Rabani, and {\'{E}}va Tardos.
\newblock {A}llocating {B}andwidth for {B}ursty {C}onnections.
\newblock {\em {SIAM} Journal on Computing}, 30(1):191--217, 2000.
\newblock \href {http://dx.doi.org/10.1137/S0097539797329142}
  {\path{doi:10.1137/S0097539797329142}}.

\bibitem{MeyersonRT13}
Adam Meyerson, Alan Roytman, and Brian Tagiku.
\newblock {O}nline {M}ultidimensional {L}oad {B}alancing.
\newblock In {\em Proceedings of the 16th APPROX}, pages 287--302, 2013.
\newblock \href {http://dx.doi.org/10.1007/978-3-642-40328-6_21}
  {\path{doi:10.1007/978-3-642-40328-6_21}}.

\bibitem{Molinaro19}
Marco Molinaro.
\newblock {S}tochastic {$\ell_p$} {L}oad {B}alancing and {M}oment {P}roblems
  via the {L}-function {M}ethod.
\newblock In {\em Proceedings of the 30th {ACM}-{SIAM} Symposium on Discrete
  Algorithms}, pages 343--354, 2019.
\newblock \href {http://dx.doi.org/10.1137/1.9781611975482.22}
  {\path{doi:10.1137/1.9781611975482.22}}.

\bibitem{Sandeep21}
Sai Sandeep.
\newblock {A}lmost {O}ptimal {I}napproximability of {M}ultidimensional
  {P}acking {P}roblems.
\newblock {\em CoRR}, abs/2101.02854, 2021.
\newblock URL: \url{https://arxiv.org/abs/2101.02854}.

\end{thebibliography}

\appendix \label{appstart}

\section{Omitted proofs from Section~\ref{sec:stochlb}} \label{append-stochlb}

We give a proof of Lemma~\ref{lem:trunctopl} in this section.
The following two technical lemmas on effective sizes will be useful to us.

\begin{lemma}[Lemma~3.8 in~\cite{IbrahimpurS20b}] \label{lem:smallbeta}
Let $\scrv$ be a nonnegative random variable and $\lambda \geq 1$. 
If $\eff{\scrv} \leq b$, then $\pr{\scrv \geq b+c} \leq \lambda^{-c}$, for any
$c \geq 0$.
Furthermore, if $\lambda \geq 2$, then $\E{\erv{\scrv}{\eff{\scrv}+1}} \leq \bigl(\eff{\scrv}+3\bigr)/\ld$. 
\end{lemma}

Lemma~\ref{lem:largebeta} below is a partial converse to Lemma~\ref{lem:smallbeta}.

\begin{lemma}[Lemma~3.9 in~\cite{IbrahimpurS20b}] \label{lem:largebeta}
Let $\ld \geq 1$ be an integer and $\scal$ be a positive scalar.
Let $\sumrv = \sum_j \scrv_j$ be a sum of independent $[0,\scal]$-bounded random
variables.
Then, 
\[
\E{\erv{\sumrv}{\scal}} \geq \frac{\scal}{4 \ld} \cdot \left(\sum_j \eff{\frac{\scrv_j}{4\scal}} - 6\right).
\]
\end{lemma}

\begin{proofof}{Lemma~\ref{lem:trunctopl}}
Recall that $\ell \in [m]$ is a positive integer, $\scal$ is a positive scalar, and $\ld := \floor{\frac{2m}{\ell}}$.
Fix an assignment $\sg$ of jobs to machines, and let $Y := \frac{1}{4\scal} \cdot \lvecsg$ denote the induced load vector scaled down by a factor $4\scal$.
For the first part, suppose that the total effective size of all jobs is at most $8m$: $\sum_{j \in J} \eff{\frac{X_j}{4\scal}} \leq 8m$.
We want to argue that the expected $\topl$ norm of the induced load vector is bounded from above.
Recall that $\al := \max\{1, \max_{i \in [m]} \sum_{j : \sg(j) = i} \eff{X_j/4\scal} \}$ is the maximum of $1$ and the makespan of $\sg$ w.r.t. $\eff{\cdot}$-sizes.
Since $X_j$s are independent random variables, for any $i \in [m]$ we have $\eff{Y_i} = \sum_{j: \sg(j) = i} \eff{\frac{X_j}{4\scal}} \leq \al$.
By Lemma~\ref{lem:smallbeta}:
\[
\E{\erv{Y_i}{\al+1}} \leq \E{\erv{Y_i}{\eff{Y_i}+1}} \leq \frac{\eff{Y_i}+3}{\ld} \leq \frac{\al+3}{\ld}.
\]
Summing over all $i \in [m]$, we get $\sum_{i \in [m]} \E{\erv{Y_i}{\al+1}} \leq (\al+3) \frac{m}{\ld} \leq (\al+3)\ell$, where we use that $\ld = \floor{2m/\ell} \geq m/\ell$.
Using Theorem~\ref{thm:exptopl}(i), with the scalar parameter $\al+3$, we get $\E{\Topl{Y}} \leq 2(\al+3)\ell$.
By homogeneity of norms (and noting that $\al \geq 1$), we get 
\[
\E{\Topl{\lvecsg}} = 4 \scal \E{\Topl{Y}} \leq 8 (\al+3) \ell \scal \leq 32 \cdot \al \cdot \ell \scal.
\]

For the second part, suppose that the total effective size of all jobs is larger than $8m$ i.e., $\sum_{j \in J} \eff{\frac{X_j}{4\scal}} > 8m$.
We want to argue that the expected $\topl$ norm of the induced load vector is bounded from below.
For any $i \in [m]$, $Y_i = \sum_{j: \sg(j) = i} \bigl(X_j/4\scal\bigr)$ is a sum of independent $[0,1]$-bounded random variables.
By Lemma~\ref{lem:largebeta}, we get:
\[
\sum_{i \in [m]} \E{\erv{Y_i}{1}} \geq \frac{1}{4\ld} \cdot \sum_{i \in [m]} \left(\sum_{j : \sg(j) = i} \eff{\frac{X_j}{4\scal}} - 6\right) = \frac{1}{4\ld} \cdot \left(\sum_{j \in J} \eff{\frac{X_j}{4\scal}} - 6m \right) > \frac{(8m-6m)}{4\cdot \frac{2m}{\ell}} = \ell/4,
\]
where we use $\ld \geq 2m/\ell$.
Using Theorem~\ref{thm:exptopl}(ii), with the scalar parameter $1/4$, we get $\E{\Topl{Y}} > \ell/8$.
Subsequently, $\E{\Topl{\lvecsg}} = 4 \scal \E{\Topl{Y}} > \ell \scal/2$.
\end{proofof}

\end{document}